\documentclass[letterpaper,twocolumn,10pt]{article}
\usepackage{usenix2019_v3}

\usepackage{tikz}
\usepackage{amsmath}
\usepackage{amsfonts}
\usepackage{amssymb}
\usepackage{amsthm}
\usepackage[noend]{algpseudocode}
\usepackage{algorithm}
\usepackage{mdframed}
\usepackage{booktabs}
\usepackage{multirow}
\usepackage{subfigure}
\usepackage{hyperref}
\usepackage{cleveref}

\usepackage{filecontents}

\newtheorem{definition}{Definition}
\newtheorem{theorem}{Theorem}
\newtheorem{corollary}{Corollary}
\newtheorem{proposition}{Proposition}
 
\newcommand{\sys}{$\mathsf{GReDP}$}
\newcommand{\vect}[1]{\ensuremath{\mathbf{#1}}}

\begin{document}

\date{}

\title{\Large \bf Training with Differential Privacy: A Gradient-Preserving Noise Reduction Approach with Provable Security}

\author{
{\rm Haodi Wang}\\
City University of Hong Kong\\
Lab for AIFT
\and
{\rm Tangyu Jiang}\\
Tsinghua University
 \and
 {\rm Yu Guo}\\
Beijing Normal University
\and
{\rm Chengjun Cai}\\
City University of Hong Kong\\
(Dongguan)
\and
{\rm Cong Wang}\\
City University of Hong Kong
\and
{\rm Xiaohua Jia$^*$}\\
City University of Hong Kong
} 

\maketitle

\begin{abstract}
Deep learning models have been extensively adopted in various regions due to their ability to represent hierarchical features, which highly rely on the training set and procedures. 
Thus, protecting the training process and deep learning algorithms is paramount in privacy preservation. 
Although Differential Privacy (DP) as a powerful cryptographic primitive has achieved satisfying results in deep learning training, the existing schemes still fall short in preserving model utility, i.e., they either invoke a high noise scale or inevitably harm the original gradients.
To address the above issues, in this paper, we present a more robust and provably secure approach for differentially private training called \sys. 
Specifically, we compute the model gradients in the frequency domain and adopt a new approach to reduce the noise level. 
Unlike previous work, our \sys~only requires half of the noise scale compared to DPSGD \cite{dpsgd} while keeping all the gradient information intact. We present a detailed analysis of our method both theoretically and empirically. The experimental results show that our \sys~works consistently better than the baselines on all models and training settings. 
\end{abstract}


\section{Introduction}
Deep learning has been widely adopted in various applications due to its ability to represent hierarchical features. The performance of the deep learning models relies highly on the training data, which usually contain sensitive and private information. A plethora of techniques have proven that an adversary can extract or infer the training data via the model parameters and training procedures \cite{Wang_2023_CVPR, 10.1145/3538707, qian2}. Consequently, protecting deep learning algorithms is paramount in preserving data privacy. 

Meanwhile, Differential Privacy (DP), as one of the fundamental cryptographic tools, offers a robust mechanism to safeguard sensitive information while extracting meaningful insights. By injecting carefully calibrated noise into the data, DP ensures that statistical queries about a dataset remain indistinguishable whether any single data point is included or excluded. 
In the context of deep learning, implementing DP on models benefits by mitigating the risks of leaking information about individual data points. This approach not only enhances the security of deep learning applications but also fosters greater trust and compliance with privacy regulations in an increasingly data-driven world.

However, directly adopting DP to the well-trained model parameters would devastate the model utility. Due to the vulnerability of deep learning models, a slight change in the model parameters could cause a divergent output \cite{sun2021exploring}, which necessitates more sophisticated approaches to preserve privacy. In 2016, Abadi \emph{et al.} \cite{dpsgd} proposed DPSGD, which modifies the standard training procedure by incorporating DP into the gradient computation process. Instead of adding noise directly to the model parameters, DPSGD gradually introduces noise to the gradients during training. There are quantities of subsequent work established on DPSGD to enhance its performance \cite{10061732, WANG2023408, ren2}. Though effective, these works still suffer from the high utility loss of the deep learning models. 

To address the above issue, in 2023, Feng \emph{et al.} proposed an alternative to DPSGD called Spectral-DP \cite{spectral}. The authors point out that the utility loss in DPSGD and the variants derive from the direct gradient clipping and noise addition to the model. Building atop of the observation that the vectors have more sparse distributions in the transferred domain \cite{FFT1, xiao2010differential}, Spectral-DP designs a novel method to perform the low-bandwidth noise addition on the model gradients using Fast Fourier Transformation (FFT). 
The sparsity enables Spectral-DP to realize a noise reduction via a filtering operation. 
Although this method significantly enhances the model utility compared to the previous work, the noise reduction in Spectral-DP highly relies on a filtering ratio $\rho$. This hand-crafted hyperparameter results in an unstable noise level. Moreover, the reduction method in Spectral-DP inevitably discards parts of the original gradients along with the noise thus the performance can be further enhanced. 

Generally speaking, the model utility under different DP methods is determined by two primary factors: (i) how much noise is added to the model gradients, and (ii) how many original gradients are retained. Analyzed from this perspective, the existing schemes either invoke a high noise scale or might cause information loss during the algorithm, thus leading to an unsatisfying model accuracy. 

\noindent\textbf{Our Objectives}. The objective of this paper is to shed some light on the differentially private deep learning training. Our proposed algorithm should retain all the gradient information and add less noise than the current schemes, thus achieving better model utility than the baseline methods with a more comprehensive theoretical analysis. 

\noindent\textbf{Our Contributions.} To this end, we propose a novel and more robust DP approach for the training procedure of the deep learning models called \sys~(i.e., DP with \underline{G}radient-preserving noise \underline{Re}duction). 
Enlightened by \cite{FFT1, xiao2010differential}, we compute the model gradients in the frequency domain by adopting the FFT and Gaussian noise mechanism to the model parameters and feature maps. 
Since the model weights are more sparse in the transferred domain, it is viable for us to perform the noise reduction. Unlike the previous work that filters the noised gradients in the frequency domain \cite{spectral}, we instead process them in the time domain by deleting the imaginary parts. We theoretically prove that our \sys~requires only half of the noise level than DPSGD under the same privacy budget. Compared to Spectral-DP, our approach invokes lower and more stable noise levels without deleting any gradient information. 
We fully implement our \sys~and design comprehensive experiments to evaluate its performance. We adopt five universally used architectures to show the performance gain under two popular training settings. The experimental results demonstrate that \sys~achieves consistently better model utility than the existing baselines under the same privacy requirements.
%
%
In all, the contributions of this paper are as follows.

\begin{itemize}
    \item We theoretically analyze the noise scale needed to maintain a certain privacy requirement. By proposing a new gradient-preserving noise reduction scheme,
    we provide a tighter bound for differentially private training.
	
   \item We propose a new DP scheme for the deep learning training procedure called \sys. Compared to the existing works, our algorithm requires the lowest noise level without discarding any gradient information under the same privacy budget, thus leading to the best model utility.
	
   \item We extensively implement our algorithm and design comprehensive experiments to evaluate our method on popular models and multiple datasets. The experimental results show that our method achieves the highest model accuracy on \emph{all} settings under the same privacy budget. 
\end{itemize}


\section{Preliminaries}

%
We denote $\mathbb{N}$ the field of natural number. We use bold letters to denote vectors and matrices, respectively. For instance, $\vect{A}$ is a vector with entry $a_i$. We adopt $[1, N]$ to represent $\{1, 2, \cdots, N\}$.

\begin{figure*}
    \centering
    \includegraphics[scale=0.6]{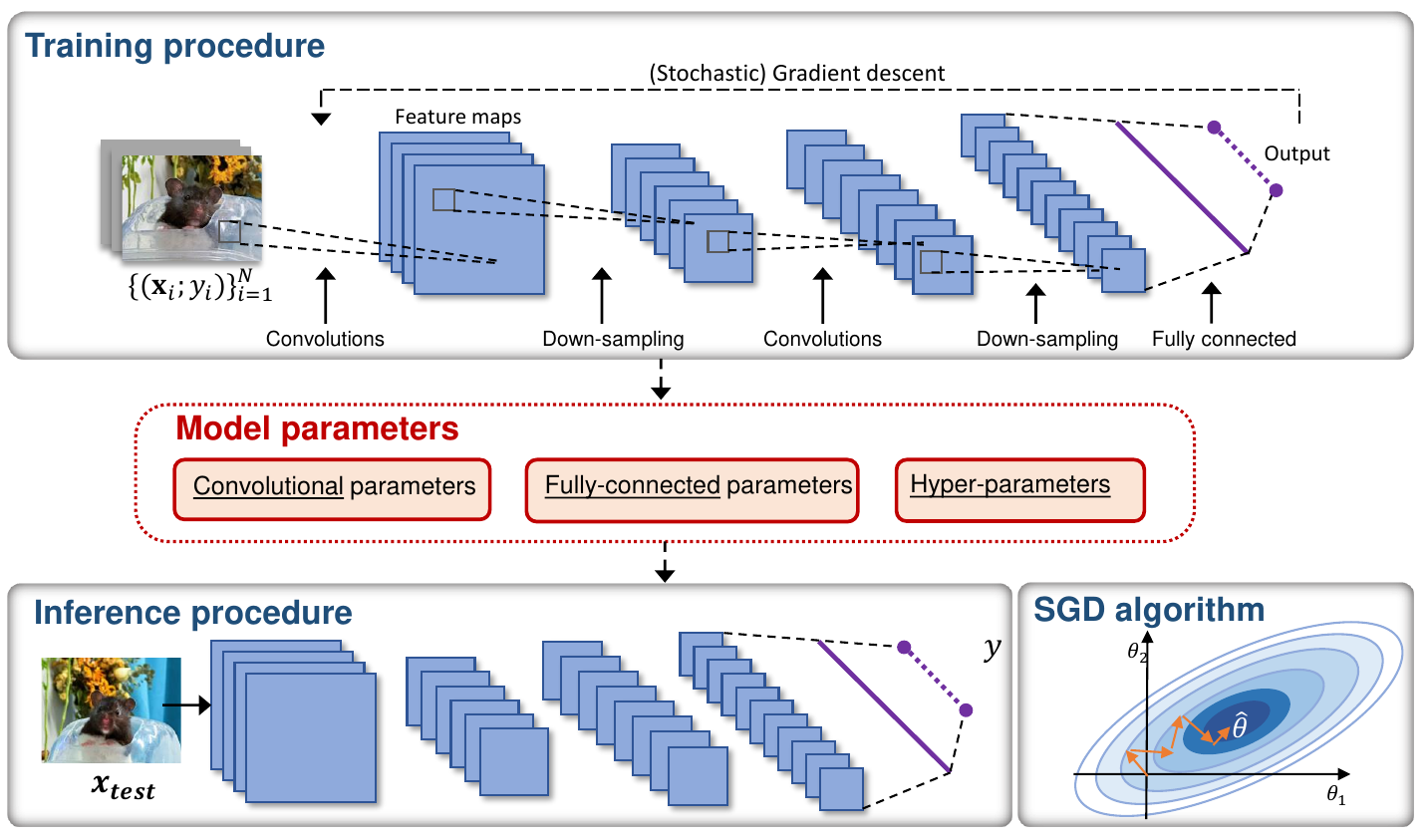}
    \caption{A typical deep learning procedure and SGD algorithm.}
    \label{fig:deeplearning}
\end{figure*}

\subsection{Differential Privacy}

DP provides mathematical privacy guarantees for specific algorithms by computing statistics about the inputs or algorithms. The goal of the original DP is to protect the membership information of each individual's data. Specifically, by observing the output, an adversary cannot distinguish whether a particular data is included in the dataset or not. As a critical primitive for privacy preservation, DP provides an essential measure to quantify privacy for various tasks. Formally speaking, for two collections of records $\mathbf{x}, \mathbf{y} \in \mathbb{N}^\chi$ from a universe $\chi$, the distance between $\mathbf{x}$ and $\mathbf{y}$ is defined such that

\begin{definition} (Distance Between Databases).
    The distance between two databases $\mathbf{x}$ and $\mathbf{y}$ is the $l_1$ distance between these two databases $\| \mathbf{x}-\mathbf{y}\|_1$, where the $l_1$ norm is:
    \begin{equation}
        \|\mathbf{x}\|_1 = \sum_{i=1}^{|\chi|} |\mathbf{x}_i|
    \end{equation}
\end{definition}

Note that $\| \mathbf{x}-\mathbf{y}\|_1$ evaluates the amount of records that differ between $\mathbf{x}$ and $\mathbf{y}$. Based on this definition, DP is defined as below. 

\begin{definition} (Differential Privacy)
    A randomized algorithm $\mathcal{M}$ with domain $\mathbb{N}^{|\chi|}$ is said to be $(\epsilon, 
    \delta)$-differentially private if for all $\mathcal{S} \subseteq \mathsf{Range}(\mathcal{M})$ and for all $\mathbf{x}, \mathbf{y} \in \mathbb{N}^{|\chi|}$ such that $\|\mathbf{x}-\mathbf{y}\|_1 \leq 1$:

    \begin{equation}
        \mathsf{Pr} \left[\mathcal{M}(\mathbf{x}) \in \mathcal{S} \right] \leq \mathsf{exp}(\epsilon) \mathsf{Pr} \left[ \mathcal{M}(\mathbf{y}) \in \mathcal{S}\right] + \delta
    \end{equation}
    where the probability space is over the coin flips of the mechanism $\mathcal{M}$.
\end{definition}

Note that if $\delta=0$, $\mathcal{M}$ is said to be $\epsilon$-differentially private. Generally speaking, $(\epsilon, \delta)$-DP ensures that for all adjacent $\mathbf{x}$ and $\mathbf{y}$, the absolute value of the privacy loss will be bounded by $\epsilon$ with probability at least $1-\delta$. 

In order to create algorithms that satisfy the definition of $(\epsilon, \delta)$-DP, one can exert the Gaussian mechanism to the target database. The noise level is determined by the Gaussian noise added to the database, which is controlled by a parameter $\sigma$. In particular, to approximate a deterministic function $f: \mathcal{M} \rightarrow \mathbb{R}$ with the $(\epsilon, \delta)$-DP is to add the Gaussian noise defined as

\begin{equation}\label{eq:gaussian}
    \mathcal{M}(\mathbf{x}) \triangleq f(\mathbf{x}) + \mathcal{N}(0, S^2 \cdot \sigma^2)
\end{equation}
where $\mathcal{N}(0, S^2 \cdot \sigma^2)$ is the Gaussian distribution with a mean value of 0 and standard deviation as $S^2\cdot \sigma^2$. $S$ denotes the sensitivity of function $f$. Based on the analysis of \cite{dwork2014algorithmic}, for any $\epsilon$ and $\delta \in (0, 1)$, the Gaussian mechanism satisfies $(\epsilon, \delta)$-DP if the following equation holds,

\begin{equation}
    \sigma = \frac{\sqrt{2\log (1.25/\delta)}}{\epsilon}
\end{equation}

\subsection{Deep Learning and Gradient Descent}

Deep learning models are extensively researched due to their ability in hierarchical learning. A typical deep learning structure consists of two main procedures, i.e., training and inference, as shown in Fig. \ref{fig:deeplearning}. For a specific task, the model owner uses its private dataset as the input of the deep learning model. The model parameters are initialized by randomness (i.e., training from scratch) or pre-trained model parameters (i.e., transfer learning or fine-tuning). In each iteration of the training procedure, the deep learning model updates its current parameters via backpropagation, which uses gradient descent algorithms to calculate the errors between the prediction results and data labels. 
Then the algorithm adjusts the weights and biases of each layer. The model parameters are fixed once the training procedure is completed, which is adopted by the inference phase using the testing dataset. 

Among all the gradient descent algorithms, the Stochastic Gradient Descent (SGD) approach is widely adopted due to its efficiency in finding the optimal parameters. Specifically, suppose the deep learning model adopts $\mathcal{L}(\theta)$ as the loss function with parameters $\theta$. In the backpropagation, the loss function over a batch of training samples $\{\mathbf{x}_1, \mathbf{x}_2, ..., \mathbf{x}_B\}$ is $\mathcal{L}(\theta) = \frac{1}{B}\Sigma_i \mathcal{L}(\theta, \mathbf{x}_i)$, where $B$ is the batch size. The optimizing goal is to minimize $\mathcal{L}(\theta)$ consistently. To this end, the model calculates the partial derivative $\nabla_\theta \mathcal{L}(\theta) := \frac{\partial \mathcal{L}}{\partial \theta}$ and updates the model parameters such that

\begin{equation}\label{eq:SGD}
    \theta_{i+1} = \theta_i - \eta \times \nabla_\theta \mathcal{L}(\theta)
\end{equation}
where $i$ indicates the training iteration, $\eta$ is the learning rate.

\subsection{Differential Privacy for Deep Learning}

There are different types of DP for deep learning models based on various privacy-preserving goals. In this paper, we adopt DP to the training process of the models. Consequently, the training procedure as well as the trainable parameters, will not leak the information of the training datasets. In particular, we require the algorithm that optimizes the deep learning model to satisfy $(\epsilon, \delta)$-DP w.r.t. the training data. To this end, the Gaussian noise is added to the clipped gradients in each iteration. The overall privacy-preserving level is determined by the privacy accountant mechanism introduced in \cite{renyi}. 
Specifically, we follow the definitions of R$\acute{e}$nyi Differential Privacy (RDP) \cite{renyi} used in the previous work \cite{dpsgd, spectral}, such that

\begin{definition} (R$\acute{e}$nyi Differential Privacy)\label{def:renyi}
    Suppose $\mathcal{M}$ is a randomized algorithm with domain $\mathcal{D}$, and $\mathcal{R} = \mathsf{Range}(\mathcal{M})$ is the range of $\mathcal{M}$. $\mathcal{M}$ is said to be $(\epsilon, \delta)$-RDP if for any two adjacent sets $d, d'\in \mathcal{D}$, the following holds
    \begin{equation}
        \mathsf{D}_{\alpha}(\mathcal{M}(d)\| \mathcal{M}(d')) \leq \epsilon
    \end{equation}
    where $\mathsf{D}_{\alpha}(\cdot \| \cdot)$ is the R$\acute{e}$nyi divergence between two probability distributions. 
\end{definition}

The RDP mechanism can be converted to a DP mechanism, which is described in the following proposition. 

\begin{proposition}\label{pro:concate}
    If $f$ is an $(\alpha, \epsilon)$-RDP mechanism, it equivalently satisfies $(\epsilon + \frac{\log 1/\delta}{\alpha-1}, \delta)$-DP for any $\delta \in (0, 1)$.
\end{proposition}

\subsection{DPSGD and Spectral-DP}

Among all the previous proposed methods, DPSGD and Spectral-DP are the most relative ones to our \sys. We recall their conceptual foundations first to better demonstrate the design and differences of our scheme.

\noindent\textbf{DPSGD \cite{dpsgd}}. DPSGD, proposed in 2016 by Abadi \emph{et al.}, is the first work that proposed to craft the synergy of DP into the region of gradient descent algorithms. In contrast to directly adding noise to the well-trained model parameters, the algorithm in \cite{dpsgd} can preserve a better model utility under more relaxed security assumptions. The method of DPSGD is relatively straightforward. Specifically, in each iteration of the training, the algorithm clips the gradient and adds the Gaussian noise to the gradients, such that

\begin{equation}
    \vect{\tilde{g}}_t = \frac{1}{B}(\sum_i \vect{g}_t(\vect{x_i}) + \mathcal{N}(0, \sigma^2 S^2))
\end{equation}
where $B$ is the batch size and $\vect{g}$ is the weight gradients. It can be demonstrated that the noise is added to the clipped gradients directly in each iteration in the time domain. Nevertheless, DPSGD is proven to be able to maintain the model utility to a certain extent efficiently. 

\noindent\textbf{Spectral-DP \cite{spectral}.} 
Feng \emph{et al.} propose that the loss in model accuracy is due to the direct gradient clipping and noise addition. Hench the authors design an alternative of DPSGD called Spectral-DP \cite{spectral}.
In particular, Spectral-DP first pads the weight matrix and the convolutional kernels to the same size and performs the Hadamard product between these two matrices in the Fourier domain. The clipping and Gaussian noise adding are applied accordingly afterward. At the core of Spectral-DP lies a filtering operation, which sets a ratio of the gradients to zeros before converting the noised gradients back to the time domain. Compared to DPSGD, Spectral-DP achieves a lower noise scale. However, it relies on the filtering ratio and would inevitably delete the original gradients. We will further analyze this in Section \ref{subsec:method:theoretical}.


\section{Method}

In this section, we present the detailed designs of \sys.
We first describe our proposed approach in general and give detailed designs to adapt our method to the deep learning models. We then analyze the theoretical results of our \sys, to explain why \sys~effective and how the performance of our approach in theory. We also elaborate on the advantage of \sys~over the existing baseline methods \cite{dpsgd, spectral}. Finally, we present the complete training algorithm of our approach. 

\subsection{Our Approach}
\label{subsec:method:approach}

\subsubsection{Overview}
\label{subsubsec:method:approach:overview}

\begin{figure*}
    \centering
    \includegraphics[scale=0.5]{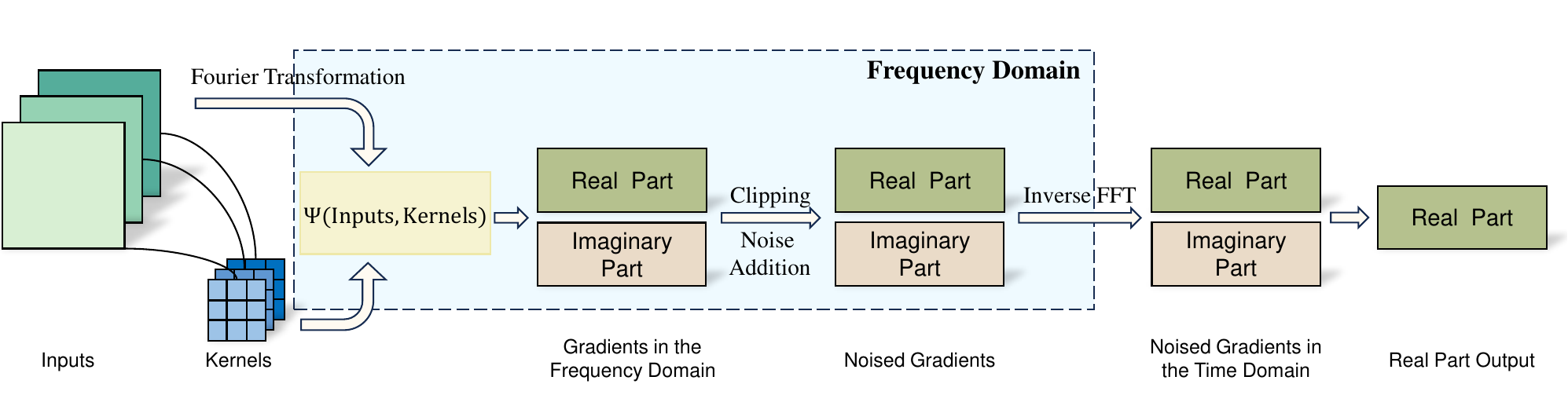}
    \caption{The overview of \sys.}
    \label{fig:overview}
\end{figure*}

The framework of \sys~is shown in Figure \ref{fig:overview}.
There are in total four steps contained in \sys, i.e., the Gradients Computation, Clipping and Noise Addition, Inverse Transformation, and Real Part Selection. In each round of the backpropagation, the gradients of each layer go through the entire procedure with noise added. 

\noindent\textbf{Gradients Computation}. The first step is to convert the gradient computation into the frequency domain. To this end, we adopt the Fourier transformation to the inputs and kernels respectively, and obtain their spectral representations. Without loss of generality, the operations in the time domain are then equivalent to an element-wise multiplication.
We denote $\mathbf{G} := \{g_1, g_2, \cdots, g_N\}$ the gradients computed in the frequency domain, such that 
%
\begin{equation}
    \mathbf{G} = \Psi (\mathcal{F}(\mathbf{X}), \mathcal{F}(\mathbf{W}))
\end{equation}
where $\mathbf{X}$ is the input matrix, $\mathbf{W}$ is the kernel, and $\Psi$ denotes the concrete operation determined by the layer type. $\mathcal{F}$ is the function of Fourier transformation such that

\begin{equation}
    \mathcal{F}(\mathbf{Y}): \hat{y}_i = \frac{1}{\sqrt{N}}\sum_{n=0}^{N-1} y_n e^{-j \frac{2\pi}{N}in}
\end{equation}
where $\hat{y}_i$ is the $i$-th term of $\mathcal{F}(\mathbf{Y})\in \mathbb{C}^N$, and $j$ is referred as the iota of the complex number.
We will defer the concrete implementation of $\Psi(\cdot, \cdot)$ for various layers in the deep learning models in the next subsection. 

\noindent\textbf{Clipping and Noise Addition}. 
The second step of \sys~is to clip the gradients and add noise to meet the privacy-preserving requirement. We follow the previous settings and inject noise layer-wise with the training procedure for better model utility. 
More concretely, we clip the $l_2$ norm of the gradients by a clipping parameter $c$ such that

\begin{equation}
    g'_i = \frac{g_i}{\max\{1, \frac{||\mathbf{G}||_2}{c} \}}
\end{equation}
where $i \in [1,N]$ and $||\cdot||_2$ is the $l_2$ normalization. The gradient clipping benefits our \sys~to control the noise sensitivity considering the scale of the noise to be added is proportional to the $l_2$ norm of the data sequence.
Then we inject the noise into the clipped gradients such that

\begin{equation}
    \hat{g}_i = g'_i + \tau_i
\end{equation}
where $\tau_i, i\in [1,N]$ is the noise drawn from $\mathcal{N}(0, \sigma^2 c^2)$ independently. 

\noindent\textbf{Inverse Transformation}. The next step is to convert the noised gradients in the frequency domain back into the time domain. To this end, we adopt the Inverse FFT (IFFT) $\mathcal{F}^{-1}(\cdot)$ to the noised gradients $\hat{\mathbf{G}}$, which is the inverse procedure of $\mathcal{F}(\cdot)$. More concretely, $\mathcal{F}^{-1}(\cdot)$ calculates the input sequence in the frequency domain such that

\begin{equation}
    \mathcal{F}^{-1}(\hat{\mathbf{Y}}): y_n = \frac{1}{\sqrt{N}}\sum_{i=0}^{N-1}\hat{y}_i \cdot e^{j\frac{2\pi}{N}in}
\end{equation}

We denote $\bar{\mathbf{G}}$ as the output of this step.

\noindent\textbf{Real Part Selection}. Finally, we proceed with the converted gradients in the time domain by deleting the imaginary part of the values and only retaining the real part. This procedure is the key operation of \sys, and also is the primary difference between our method and previous works. 
Specifically, since we add the noise in the frequency domain, the original real values (i.e., the weights and inputs in the time domain) might become complex numbers with imaginary parts. Thus, the results of IFFT also contain imaginary parts. For each of the value in $\bar{\mathbf{G}}$, we discard the imaginary value such that
\begin{equation}
    \mathcal{R}(\bar{\mathbf{G}}): \bar{g_i} := \bar{g}_i^\mathsf{r} + j\cdot \bar{g}_i^{\mathsf{im}} \rightarrow \tilde{g}_i := \bar{g}_i^\mathsf{r}
\end{equation} 
where $i \in [1,N]$. Note that the real part selection in our \sys~is the key mechanism that can effectively decrease the influence of the injected noise on the model utility. Moreover, the way that we only reserve the real value of the gradients can retain the entire gradient information than the existing approaches. We will analyze this in detail in the next subsection. 
The overall conceptual procedure of \sys~is shown in Algorithm \ref{alg:sys}. In the following sections, we denote $\mathsf{GRe}(\mathbf{G})$ as the procedure of inputting the gradients $\mathbf{G}$ and obtaining $\tilde{\mathbf{G}}$ as output.

\subsubsection{\sys~in Deep Learning}

Having the conceptual procedure of \sys, we can construct concrete ways to adapt our approach in the deep learning models (i.e., instantiate $\Psi$).
The key question is how to compute the model gradients in the frequency domain. 
To answer this, we construct the methods of \sys~for the convolutional layers and fully-connected layers, respectively.

\begin{algorithm}
		\caption{Overall algorithm of \sys} 
		\label{alg:sys}
		\begin{algorithmic}[1]
			\Require
			Kernel Weight $\mathbf{W}$, input matrix $\mathbf{X}$, $l_2$ sensitivity $c$, and noise scale $\sigma$.   
			\Ensure
			Vector $\tilde{\mathbf{G}}$ with noise. 
   \State Gradients Computation: Calculate gradients $\mathbf{G}$ using backpropagation in frequency domain.
   \State Clipping and Noise Addition: $\hat{\mathbf{G}}\leftarrow \frac{\mathbf{G}}{\max \{1,\|\mathbf{G}\|_2/c\}}+\mathcal{N}(0,c^2\sigma^2)$
   \State Inverse Transformation: $\bar{\mathbf{G}}=\mathcal{F}^{-1}(\hat{\mathbf{G}})$
   \State Real Part Selection: $\tilde{\mathbf{G}}=\mathcal{R}(\bar{\mathbf{G}})$
   \end{algorithmic}
   \end{algorithm}

\noindent\textbf{Convolutional Layers}.
The convolutional operations are the most fundamental computations in deep learning. To adopt \sys~to the convolutions, we first construct the method for the two-dimensional convolutions. More concretely, suppose the input of the convolutional layer is $\mathbf{X}\in \mathbb{R}^{h\times w \times c_{\mathsf{in}}}$, and the kernel matrix is denoted as $\mathbf{W}\in \mathbb{R}^{c_{\mathsf{in}} \times c_{\mathsf{out}} \times d \times d}$, where $h, w$ are the input size, $d$ is the kernel size, and $c_{\mathsf{in}}$ (\textit{resp.} $c_{\mathsf{out}}$) is the input (\textit{resp.} output) channel of the current layer.  
Then the computation of the 2D convolution is 
\begin{equation}
    \mathbf{o}_{i} = \sum_{j}^{c_{\mathsf{in}}} \mathbf{X}_j \circledast \mathbf{W}_{i,j}
\end{equation}
where $\mathbf{o}_i \in \mathbb{R}^{h_{\mathsf{out}}\times w_{\mathsf{out}}}$ is the feature map generated in the $i$-th channel of the convolution, $\mathbf{X}_j$ denotes the $j$-th channel of the input, $\mathbf{W}_{i,j}$ is the $(i, j)$-th kernel, and $\circledast$ is the convolutional operation. The backpropagation of the 2D convolution in the time domain is equivalent to the Hadarmard multiplication between the input matrix 
$\mathbf{X}$ and convolutional kernels $\mathbf{W}$:

\begin{equation}
    \frac{\partial \mathcal{L}}{\partial \mathbf{W}_{i,j}} = \frac{\partial \mathcal{L}}{\partial \mathbf{o}_{i}} \circledast \frac{\partial \mathbf{o}_{i}}{\partial \mathbf{W}_{i,j}} = \mathcal{F}^{-1}(\mathcal{F}(\frac{\partial \mathcal{L}}{\partial \mathbf{o}_{i}})\odot \mathcal{F}(\mathbf{X}_j))
\end{equation}
where $\circledast$ is the convolutional operation in the time domain and $\odot$ denotes the element-wise multiplication. $\frac{\partial\mathcal{L}}{\partial \mathbf{W}_{i,j}}$ is the gradient of the loss function over the model weights, $\frac{\partial\mathcal{L}}{\partial \mathbf{o}_{i}}$ is the gradient of the loss function over the feature map, and $\frac{\partial \mathbf{o}_{i}}{\partial \mathbf{W}_{i,j}}$ is the gradient of the feature map w.r.t. the kernel weight matrix. Thus, \sys~can be adapted to the 2D-convolutions 
such that

\begin{equation}
    \mathsf{GRe}(\mathbf{G}_{\mathsf{conv}}) = \mathsf{GRe}( \mathcal{F}(\frac{\partial \mathcal{L}}{\partial \mathbf{o}_{i,j}})\odot \mathcal{F}(\mathbf{X}_j))
\end{equation}

\noindent\textbf{Fully-Connected Layers}. Compared to the convolutional layers, applying \sys~to the fully-connected layers in the deep learning model is less straightforward. The primary reason is that there is no direct Fourier transformation for the fully-connected operations. To alleviate this issue, we adopt the 
circular weight matrix \cite{spectral, ding2017circnn, lin2018fft} and convert the weights in the fully connected layers to a circular shifted matrix. 
We briefly present the details for comprehensiveness. 

Specifically, for a block-circulant matrix $\mathbf{W}\in\mathbb{R}^{m\times n}$, which consists of circulant square sub-matrices $\mathbf{W}_{i,j}\in\mathbb{R}^{d\times d}$, the forward propagation is given as follows.

\begin{equation}
    \mathbf{O} = \mathbf{W}\mathbf{X} = 
    \left[\begin{array}{c}
    \mathbf{o}_1\\
    \mathbf{o}_2\\
    \cdots\\
    \mathbf{o}_p
    \end{array}\right]=
    \left[\begin{array}{c}
    \sum_{j=1}^q \mathbf{W}_{1, j}\mathbf{X}\\
    \sum_{j=1}^q \mathbf{W}_{2, j}\mathbf{X}\\
    \cdots\\
    \sum_{j=1}^q \mathbf{W}_{p, j}\mathbf{X}
    \end{array}\right]
\end{equation}
where $p=m\div d$, $q=n\div d$ and each subsequent row of $\mathbf{W}_{i,j}$ is a one-position circular shift of the previous one, such that

\begin{equation}
    \mathbf{W}_{i,j} = \left[\begin{array}{cccc}
       w_0  &  w_1 & \cdots & w_{d-1}\\
       w_{d-1}  &  w_0 & \cdots & w_{d-2}\\
       \cdots \\
       w_1  &  w_2 & \cdots & w_0 \\
    \end{array}
    \right]
\end{equation}

According to the circulant convolution theorem \cite{pan2012structured, bini2012polynomial}, $\mathbf{W}_{i,j}\mathbf{X}$ can be further equivalently represented as $\mathcal{F}^{-1}(\mathcal{F}(\mathbf{w}_{i,j})\odot\mathcal{F}(\mathbf{X}))$, where $\mathbf{w}_{i,j}$ is the first row vector of $\mathbf{W}_{i,j}$.
Then using the chain rule, the backpropagation of the fully connected layers can be derived such that

\begin{equation}\label{eq: fully-back}
\begin{split}
    \frac{\partial \mathcal{L}}{\partial \mathbf{w}_{i,j}} &= \sum_{l=1}^N \frac{\partial \mathcal{L}}{\partial \mathbf{o}_{i,l}}\frac{\partial \mathbf{o}_{i,l}}{\partial \mathbf{w}_{i,j}}
    =\frac{\partial \mathcal{L}}{\partial \mathbf{o}_i}\frac{\partial \mathbf{o}_i}{\partial \mathbf{w}_{i,j}}\\
    &=\mathcal{F}^{-1}(\mathcal{F}(\frac{\partial \mathcal{L}}{\partial \mathbf{o}_i})\odot \mathcal{F}(\frac{\partial \mathbf{o}_i}{\partial \mathbf{w}_{i,j}}))
    \end{split}
\end{equation}
where $\partial \mathbf{o}_i/\partial \mathbf{w}_{i,j}$ is block-circulant matrix.
Thus, the \sys~can be adopted to the fully connected layers with a block-circulant matrix such that 

\begin{equation}
    \mathsf{GRe}(\mathbf{G}_{\mathsf{fc}}) = \mathsf{GRe}( \mathcal{F}(\frac{\partial \mathcal{L}}{\partial \mathbf{o}_{i}})\odot \mathcal{F}(\frac{\partial \mathbf{o}_i}{\partial \mathbf{w}_{i,j}}))
\end{equation}

\subsection{Theoretical Analysis}
\label{subsec:method:theoretical}

In this subsection, we analyze the theoretical results of our proposed \sys. We first illustrate the advantage of \sys~over the existing methods, and prove the contributions in theory. Then, in Theorem \ref{theo:dpsingle} and Corollary \ref{coro:conv}, we analyze the noise level \sys~achieves.

At first glance, \sys~is similar to Spectral-DP \cite{spectral} since they all utilize the FFT technique to the model gradients and complete the Gaussian mechanism in the frequency domain. The primary differences between \sys~and Spectral-DP consist two points: 1) we first perform IFFT to the noised gradients and then ``filter'' the gradients in the time domain, while Spectral-DP first filters the gradients in the frequency domain and then adopts the IFFT; and 2) we perform the noise reduction by selecting the real part of the values while Spectral-DP filters the noised gradients by setting parts of the vectors to zeros with a filtering ratio $\rho$. 
Note that those differences are not marginal improvements to the previous work but provide much better privacy and performance gain for \sys. More concretely, we have the following key theorem. 

\begin{theorem}\label{theo:main}
    Suppose the noise is $\tau = \tau_x + i \cdot \tau_y$, where $\tau_x, \tau_y \sim \mathcal{N}(0, \frac{1}{2} S^2 \sigma^2)$. 
    Denote a sequence of values as a vector $\mathbf{V}=(V_1, \cdots, V_N)$,
    then the output $\mathsf{GRe}(\mathcal{F}(\mathbf{V}))$ satisfies:

    \begin{equation}
        \mathsf{GRe}(\mathcal{F}(\mathbf{V}))=\mathcal{R}(\mathcal{F}^{-1}(\mathcal{F}(\mathbf{V})+\tau))\sim \mathbf{V}+\mathcal{N}(0, \frac{1}{2} S^2\sigma^2)
    \end{equation}
\end{theorem}

\begin{proof}
    Let the vector as $\mathbf{V}=(V_1, \cdots, V_N)$, and $\mathbf{T}= \mathcal{F}(\mathbf{V})$, such that $T_j=\frac{1}{\sqrt{N}}\sum_{k=1}^N V_k\exp(-jki2\pi/N)$. Then the noise addition as described in Section \ref{subsec:method:approach} is $\hat{\mathbf{T}}=\mathbf{T}+\tau$. Thus, performing the IFFT on the noised vector in the frequency domain can be represented as 

    \begin{equation}
    \begin{split}
    \mathcal{F}^{-1}(\hat{\mathbf{T}})
			&=\mathcal{F}^{-1}(\mathbf{T})+\mathcal{F}^{-1}(\tau)=\mathbf{V}+\mathcal{F}^{-1}(\tau)
		\end{split}
    \end{equation}
    where $\mathcal{F}^{-1}(\tau)
    \sim
\mathcal{N}(0,\frac{1}{2} S^2\sigma^2)+i\mathcal{N}(0,\frac{1}{2} S^2\sigma^2)$. Thus,

\begin{equation}
    \mathsf{GRe}(\mathcal{F}(\mathbf{V})) = \mathbf{V}+\mathcal{R}(\mathcal{F}^{-1}(\tau))
    \sim \mathbf{V}+\mathcal{N}(0,\frac{1}{2} S^2 \sigma^2)
\end{equation}
which completes the proof.
\end{proof}

The results in Theorem \ref{theo:main} explain why \sys~is effective from two perspectives. Firstly, from the noise level perspective, this main theorem indicates that \textbf{\sys~can achieve the desired level of privacy requirement with lower injected Gaussian noise scale}, i.e., only half of the original standard deviation $S^2 \sigma^2$. On the contrary, DPSGD injects two times of the noise to the model gradients (i.e., $\tau\sim \mathcal{N}(0, S^2\sigma^2)$), while Spectral-DP adds noise following distribution $\mathcal{N}(0, \frac{K}{N}S^2 \sigma^2)$, where $N$ is the length of the gradient sequence and $\frac{K}{N}$ determines the fraction of coefficients that are retained in the filtering. Note that when $K=\frac{N}{2}$ (i.e., indicating the filtering ratio $\rho=0.5$), Spectral-DP has the same noise level as \sys~theoretically. However, when $K > \frac{N}{2}$ (i.e., $\rho>0.5$), Spectral-DP will add a higher scale of noise than ours. In real-world experiments, $\rho$ is a hyperparameter which differs with various model architectures. The evaluation also show that $\rho>0.5$ is a better choice under certain circumstance, making the noise level in Spectral-DP higher than ours. 
Secondly, from the gradient information perspective, Theorem \ref{theo:main} illustrates that \textbf{\sys~can retain all of the gradient information during the training procedure}. On the contrary, the filtering operation in Spectral-DP would cause the original gradients damaged. Specifically, the filtering is the key step in Spectral-DP which deletes a ratio (i.e., $\rho$) of noised gradients in the frequency domain. This procedure can effectively reduce the noise level thus improving the model utility. However, since the filtering operation is exerted on the values in the frequency domain, which are composed of the original gradients and the noise generated following the Gaussian mechanism. Thus, deleting parts of the vectors would remove both the essential (i.e., gradients) and non-essential (i.e., noise) information. In fact, Spcetral-DP managed to trade-off between reducing the noise as much as possible and preserving more gradients to enhance the model utility.
In contrast to Spectral-DP, \sys~keeps all the gradient information intact. It can be seen from Theorem \ref{theo:main} that $\mathsf{GRe}(\mathbf{V}) = \mathbf{V} + \mathcal{N}(0, \frac{1}{2}S^2 \sigma^2)$, indicating that all the original $\mathbf{V}$ is retained. The key observation is that we select a different way to perform the noise reduction. Since the noise addition occurred in the frequency domain, the imaginary part of the gradient vector \textit{contains only noise}. Therefore, we first perform the IFFT and preserve the real part to reduce the noise level. Consequently, our \sys~keeps all the original gradients which contributes to a higher accuracy. 

We can generalize the above results by relaxing the assumption of $\tau_x$ and $\tau_y$, and obtain the following theorem.

\begin{theorem}\label{theo:extendmain}
    Suppose the noise is $\tau = \tau_x + i \cdot \tau_y$. The real and imaginary part of $\tau$ satisfy: $\tau_x \sim \mathcal{N}(0, a^2 S^2 \sigma^2)$ and $\tau_y \sim \mathcal{N}(0, b^2 S^2 \sigma^2)$, where $\sqrt{a^2+b^2}=1$. Then the results in Theorem \ref{theo:main} still hold.
\end{theorem}

\noindent\textit{Proof Sketch}. The proof for Theorem \ref{theo:extendmain} is similar to Theorem \ref{theo:main}. We present the complete proof in Appendix \ref{app:proof_theo_2}.

The above theoretical results demonstrate that our scheme retains all the original gradient information while requiring only half of the noise level in the previous work, thus guaranteeing the performance. 
We then present the following theorem to analyze the DP level adopted in \sys. 

\begin{theorem}\label{theo:dpsingle}
    In Algorithm 1 (i.e., corresponding to our method overview), the output $\mathsf{GRe}(\mathbf{G})$ is $(\epsilon, \delta)$ differentially private if we choose $\sigma$ to be $\sqrt{2\log (1.25/\delta)}/\epsilon$.    
\end{theorem}

\noindent\textit{Proof Sketch.}
The proof of the above theorem derives from Theorem 3.22 in \cite{dwork2014algorithmic} and the post-processing property of the DP algorithm. One of the most essential observations is that both the $\mathsf{IFFT}$ and $\mathsf{Re}$ operations are post-processing, and will not affect the DP level. We present the details in Appendix \ref{app:proof_theo_3}.


Furthermore, to demonstrate the noise level of converting the two-dimensional convolutions to the element-wise multiplications in the frequency domain, we present the following corollary. 

\begin{corollary}\label{coro:conv}
    Suppose $\Gamma =\{\gamma_{m,n}\}_{m,n \in [0,N-1]}$ is a matrix of noise vector in frequency domain where $\gamma_{m,n} \sim \mathcal{N}(0, \sigma^2 S^2)$. Then the elements in $\mathcal{R}(\mathcal{F}^{-1}(\Gamma))$ follow a normal distribution $\mathcal{N}(0, \frac{1}{2}\sigma^2 S^2)$.
\end{corollary}

The proof for Corollary \ref{coro:conv} can be found in Appendix \ref{app:proof_coro_1}.  

\subsection{Training Algorithm}

We present the complete training algorithm of \sys~in Algorithm \ref{alg:trainingalgorithm}. In each iteration, the algorithm takes a random batch of data with sampling probability $B/N$ and calculates the gradients in the last layer. The gradients are then processed as described before and backpropagated to the previous layer. All the parameters are updated based on the learning rate $\eta$ and noised gradients. 
To demonstrate the overall DP achieved by our scheme, we present the following corollary based on the RDP privacy accountant.

\begin{corollary}\label{coro:train}
	The DL training algorithm we proposed achieves $((T_e * N/B)\epsilon + \frac{\log(1/\delta)}{\alpha-1}, \delta)$-DP if $\sigma = \sqrt{2\log(1.25/\delta)}/\epsilon'$, where $\epsilon' = \epsilon + \frac{\log(1/\delta)}{\alpha-1}$.  
\end{corollary}

\noindent\textit{Proof Sketch}. The proof of the above corollary relies on the $(\alpha, \epsilon)$-RDP Proposition. We present the complete proof in Appendix \ref{app:proof_coro_2}.

\begin{algorithm}
		\caption{Training algorithm of \sys} 
		\label{alg:trainingalgorithm}
		\begin{algorithmic}[1]
			\Require
			Training data: $\{(x_i, y_i)\}_{i=1}^{N}$, model parameters: $\mathbf{W}$, model layers: $L$, loss function: $\mathcal{L}$, learning rate: $\eta$, batch size: B, $l_2$ sensitivity $c$, and noise scale $\sigma$, total training epochs: T.   
			\Ensure
			Model with weight $\hat{\mathbf{W}}_T$. 
   \State Initialize $\mathbf{W}_0$ randomly
   \For{$t\in[1, T*N/B]$}
   \State Take a random sample $B_t$ with sampling probability $B/N$
   \For{$i \in [1, B]$}
   \For{$l \in [L, 1$]}
   \State Calculate the gradient of each sample in the $l$-th layer in the frequency domain and get $\mathbf{G}_i^l$
   \State Clipping and noise addition: $\hat{\mathbf{G}}^l_i\leftarrow \frac{\mathbf{G}_i^l}{\max(1, \|\mathbf{G}_i^l/c\|_2)}+\mathcal{N}(0,c^2\sigma^2)$
   \EndFor
      \State \textbf{end for}
   \EndFor
      \State \textbf{end for}
   \State \textbf{Gradient descent}
   \For{$l\in[1, L]$}
   \State $\hat{\mathbf{G}}^l\leftarrow\sum_{i=1}^B\hat{\mathbf{G}}^l_i$
      \State Inverse Transformation: $\bar{\mathbf{G}}^l\leftarrow\mathcal{F}^{-1}(\hat{\mathbf{G}}^l)$
   \State imaginary Part Deletion: $\tilde{\mathbf{G}}^l\leftarrow\mathcal{R}(\bar{\mathbf{G}}^l)$
   \State Update Parameters: $\mathbf{W}^l_{t}=\mathbf{W}^l_{t-1}-\eta \frac{1}{B}\tilde{\mathbf{G}}^l$

   \EndFor
      \State \textbf{end for}
   \EndFor 
      \State \textbf{end for}
   \end{algorithmic}
   \end{algorithm}

\section{Experiments}

In this section, we would like to answer the following key questions:

\begin{itemize}
    \item How is the performance of \sys~on various models when training the models from scratch? (Section \ref{subsec:exp:scratch})

    \item How is the performance of \sys~on different models when training in the transfer settings? (Section \ref{subsec:exp:transfer})

    \item What is the influence of the hyperparameters on the performance of \sys? (Section \ref{subsec:exp:ablation})
\end{itemize}

We will first describe the experimental configurations and then present the concrete results and analysis. 

\subsection{Experimental Configurations}


\noindent\textbf{Hardware.} We fully implement our \sys~in Python programming language. We evaluate \sys~on a desktop with an Inter Core i7-12700F CPU and GeForce RTX 3090. 

\noindent\textbf{Software.} We implement all the evaluating models using Python and the PyTorch platform. All the operations in the deep learning models (i.e., the convolutional and fully connected layers) are implemented based on \cite{spectral}.

\noindent\textbf{Evaluation Metrics.} To better conduct the comparison of our \sys~with previous works, we design two types of experiments for our scheme: 1) training the models from random initialization (i.e., training from scratch), and 2) training the models from pre-trained parameters (i.e., transfer learning). We adopt the final model accuracy in the inference procedure as the main evaluation metric to compare the model utility under the same privacy requirements. 


\noindent\textbf{Datasets.} We adopt three widely-used public datasets to evaluate the performance of \sys, i.e., MNIST \cite{mnist}, CIFAR-10, and CIFAR-100\cite{cifar}. MNIST is a handwritten digital dataset that contains 60,000 images for training and 10,000 images for testing. The image size is $28\time 28$. 
CIFAR-10 is an RGB image dataset with 50,000 training data and 10,000 testing data. The input size is $32 \times 32$ and there are 10 classes contained. 
CIFAR-100 is an RGB image dataset for 100 classes, which includes 500 training and 100 testing data for each class. The input size is also $32 \times 32$. 
We select these datasets according to the concrete deep learning model adopted, whose detailed structure will be demonstrated later.

\noindent\textbf{Counterpart Comparison.} We adopt DPSGD \cite{dpsgd} and Spectral-DP \cite{spectral} as our primary baseline methods. For DPSGD, we use their public implementation to evaluate our settings, which is established on the $\mathsf{opacus}$ \cite{opacus} tool. All the operations are realized through the standard $\mathsf{nn.torch}$ library.
Since the codes are not open-sourced for Spectral-DP \cite{spectral}, we implemented the method based on the information provided. 
%
%

Our code is open-sourced and available at

\begin{center}
		\href{https://anonymous.4open.science/r/FReDP-4CB2}{https://anonymous.4open.science/r/FReDP-4CB2}
\end{center}

\subsection{Training From Scratch}
\label{subsec:exp:scratch}

We first present the performance of \sys~and comparisons under the training from scratch settings. \textbf{Overall, \sys~achieves the best model utility among all the baselines on all deep learning models and settings. }
Specifically, Table \ref{tab:acc_scratch} demonstrates the results of \sys~and baseline methods at a glance. It can be observed that \sys~outperforms DPSGD and Spectral-DP in terms of model accuracy under the same experimental settings. For instance, our approach achieves 96.5\% model accuracy on LeNet-5 and MNIST dataset, which is 3.8\% and 1.2\% higher than Spectral-DP and DPSGD, respectively, when $\epsilon=1$ and $\delta=10^{-5}$. The performance gain becomes more significant on the ResNet-20, Model-3, and AlexNet models with CIFAR-10 datasets. More concretely, \sys~achieves 4.41\% to 4.67\% higher model accuracy than Spectral-DP and over 21.4\% to 36.85\% better utility than DPSGD. Note that Spectral-DP is currently the state-of-the-art method in the literature. In the following section, we will respectively describe the performance of \sys~on each deep learning model in details.   

\begin{table}[htbp]
  \caption{Comparisons of accuracy (\%) between \sys~and baselines methods in training from scratch. }
    \label{tab:acc_scratch}
\centering
\resizebox{\columnwidth}{!}{
	\begin{tabular}{l|c|c|c|c}
		\toprule
           Dataset & MNIST & \multicolumn{3}{|c}{CIFAR-10} \\
           \midrule
           Models & LeNet-5 & ResNet-20 & Model-3 & AlexNet \\
            \midrule
            Privacy budget & $(1, 10^{-5})$  & $(1.5, 10^{-5})$ & $(2, 10^{-5})$ & $(2, 10^{-5})$\\
            Plain model & 98.72 & 70.44 & 79.17 & 64.58\\
            DP-SGD \cite{dpsgd} & 92.7 & 29.3 & 46.47 & 20.54\\
            Spectral-DP \cite{spectral} & 95.3 & 54.9 & 63.20 & 52.92\\
            \midrule
            \sys~\textbf{(Ours)} & \textbf{96.5} & \textbf{59.31} & \textbf{67.87} & \textbf{57.39}\\
			\bottomrule
		\end{tabular}
  }
\end{table}

\subsubsection{\sys~on LeNet-5 and MNIST}

LeNet-5 is one of the most fundamental model architectures in deep learning \cite{lenet}. It was first proposed in 1994 and significantly encouraged the improvement of the successive models. 
The structure of LeNet-5 is shown in Table \ref{tab:arch:lenet} in Appendix \ref{app:stru:lenet5}. It contains two convolutional layers and three fully-connected layers. The results of \sys~and baseline methods on LeNet-5 under different privacy budgets are demonstrated in the left part of Figure \ref{fig:acc} (a). It can be seen from the figure that our \sys~realizes consistently higher model accuracy than the existing baselines under various privacy budget settings. For instance, when $\epsilon=1$ and $\delta=10^{-5}$, \sys~achieves 96.5\% model accuracy, while DPSGD and Spectral-DP obtain 92.7\% and 95.3\%, which is 3.8\% and 1.2\% lower than ours, respectively. 
The performance gain is relatively steady with respect to the privacy budget. Note that a lower $\epsilon$ indicates tighter requirement of privacy. Thus, it can be shown from the figure that our approach can protect the model parameters and training dynamics while maintaining certain model utilities. 
To better illustrate the training procedure of \sys, we present the variation of the model accuracy concerning the privacy budget $\epsilon$ and training epochs on LeNet-5 and MNIST when $\epsilon=2$ in the middle and right part of Figure \ref{fig:acc} (a), respectively. It can be observed that the model utility gradually becomes better when $\epsilon$ grows larger for all three approaches. However, \sys~shows a consistently better performance than Spectral-DP and DPSGD on all values of the privacy budget. For instance, $\epsilon$ is calculated around 0.31 at the beginning of the training procedure, at which \sys~achieves 78.37\% accuracy while DPSGD and Spectral-DP obtain 33.38\% and 68.78\% model accuracy, respectively. The performance gain gradually decreases when $\epsilon$ increases from 0 to 2. The relationship between the model accuracy and training epochs is similar to the privacy budget $\epsilon$. The model accuracy with all three approaches increases rapidly from epoch 0 to 10 and becomes slow after epoch 11. During the entire training procedure, \sys~achieves higher model accuracy than Spectral-DP and DPSGD.

\subsubsection{\sys~on ResNet-20 and CIFAR-10}

ResNet was a type of model architecture proposed in 2015 \cite{resnet}. To address the issue of the vanishing/exploding gradient, ResNet architecture introduces the concept called Residual blocks equipped with skip connections. ResNet-20 is one of the typical architectures that utilize the Residual blocks to enhance the model performance, which contains three convolutional layers and one fully-connected layer. The concrete model structure is presented in Table \ref{tab:arch:resnet20} in Appendix \ref{app:stru:resnet20}.

We adopt CIFAR-10 as the dataset to evaluate the performance of \sys~and all the baselines on ResNet-20. The results of the model accuracy under different DP methods are shown in the left part of Figure \ref{fig:acc} (b). \sys~achieves consistently better model utility than DPSGD and Spectral-DP under all settings of privacy budget (i.e., the red solid line in the figure). For instance, the model accuracy using \sys~is 49.06\%, 57.36\%, 59.31\%, and 58.83\% when $\epsilon=0.5, 1, 1.5$ and 2, while Spectral-DP obtains 42.91\%, 51.61\%, 54.9\%, and 55.68\%, respectively. The performance gain is more significant compared to DPSGD, which achieves 29.55\%, 31.35\%, 31.91\%, and 32.1\% model accuracy. With the target $\epsilon$ grows larger, the model utility becomes better, which conforms to the theoretical analysis. Compared to the LeNet-5 model, the performance gain is more significant than the baselines, since ResNet-20 is more complicated than LeNet-5. 
We also present the training dynamics of the ResNet-20 model in Figure \ref{fig:acc} (b) when $\epsilon=2$, in which the middle (\textit{resp.} right) part is the relationship between the model accuracy and the current privacy budget $\epsilon$ (\textit{resp.} training epochs). It can be shown from the figure that with $\epsilon$ increases from zero to 2, the model accuracy gradually grows higher with slight perturbation. When $\epsilon$ is smaller than 1, the model accuracy obtained with \sys~and Spectral-DP is similar, while \sys~outperforms Spectral-DP for up to 5\% higher when $\epsilon > 1$. On the other hand, \sys~is consistently better than DPSGD during the entire training process. 
As for the training epochs, the trend is similar to the relationship between accuracy and the current privacy budget. \sys~demonstrates better model utility than Spectral-DP after the tenth epoch and is significantly better than DPSGD through the entire training. 

\subsubsection{\sys~on Model-3 and CIFAR-10}

Among all the models we selected in this section to evaluate \sys, Model-3 is the only one that is customized. It was designed in Spectral-DP \cite{spectral} as the most complicated model in their experiments. To better illustrate the performance of \sys~compared to the baselines, we adopt Model-3 as one of the testing models. Model-3 consists of three convolutional layers and two fully connected layers, and the concrete structure is shown in Table \ref{tab:arch:model3} in Appendix \ref{app:stru:model3}.

The model accuracy with different settings of the privacy budget is shown in the left part of Figure \ref{fig:acc} (c). Similar to the previous two models, \sys~outperforms all the baseline methods on Model-3 with respect to different privacy budgets. More concretely, \sys~achieves 56.15\%, 63.93\%, 66.3\%, and 67.87\% model accuracy when $\epsilon=0.5, 1, 1.5$ and 2, respectively. These results are 5.62\%, 3.78\%, 4.11\%, and 4.67\% higher than Spectral-DP and 13.22\%, 18.15\%, 20.03\%, and 21.4\% higher than DPSGD. The performance gain proves that our \sys~retains more parameter information than other approaches while keeping the same noise scale, which conforms to the theoretical analysis. We present the training dynamics on Model-3 in the middle and right part of Figure \ref{fig:acc} (c) when the target privacy budget is $(2, 10^{-5})$. Similar to the situations on ResNet-20, the model accuracy on the validation set gradually increases with $\epsilon$ grows larger using all three methods. During the training procedure, \sys~obtains better results than Spectral-DP and DPSGD consistently under the same value of privacy budget. The performance gain over DPSGD is more significant than Spectral-DP. The curve of the former is smooth while the latter exhibits a similar perturbation as ours. Finally, the right figure illustrates the relationship between the model accuracy and the training epochs on Model-3. It can be observed from the figure that to achieve the same model accuracy, \sys~requires fewer training epochs than the baselines. For instance, DPSGD costs 24 epochs to achieve model accuracy higher than 45\%, while \sys~only needs 2. In the meantime, \sys~requires 11 epochs to obtain over 60\% model accuracy while the best baseline Spectral-DP needs 16 epochs. 

\subsubsection{\sys~on AlexNet and CIFAR-10}

Finally, we adopt AlexNet \cite{krizhevsky2012imagenet} and CIFAR-10 to evaluate the performance of our \sys~on larger models. AlexNet was proposed in 2012. As an evolutionary improvement over LeNet-5, it achieved excellent performance in the 2012 ImageNet challenge and created the prosperity of deep learning. AlexNet consists of five convolutional and three fully-connected layers. The concrete model structure is presented in Table \ref{tab:arch:alexnet} in Appendix \ref{app:stru:alex}.

We evaluate our \sys~on AlexNet with different values of $\epsilon$ and compare it with the baseline methods. The results are shown in the left part of Figure \ref{fig:acc} (d). The performance gain of \sys~over DPSGD is the most significant on AlexNet among all the selected models. For instance, our approach achieves 51.93\%, 55.8\%, 56.89\%, and 57.39\% model accuracy, while DPSGD obtains 22.86\%, 21.11\%, 20.78\%, and 20.54\% accuracy when $\epsilon=0.5, 1, 1.5$, and 2, respectively. Our advantage over DPSGD demonstrates that \sys~can adopts a lower noise scale than DPSGD and the performance gap becomes larger on more complicated models. We also acquire consistently better results than Spectral-DP with 2.5\% to 4.47\% higher accuracy.
The relationship between the model accuracy and the current privacy budget $\epsilon$ on AlexNet when target $\epsilon=2$ is shown in the middle part of Figure \ref{fig:acc} (d). Compared to ResNet-20 and Model-3, the accuracy of the validation set is more stable and gradually increases when $\epsilon$ grows larger. 
We also illustrate the relationship between the model accuracy and training epochs in the right part of Figure \ref{fig:acc} (d). It can be demonstrated from the figure that \sys~and Spectral-DP share similar training dynamics, where the model accuracy gradually increases with slight perturbation during the training. On the other hand, the accuracy of DPSGD first increases from epoch 0 to 15, and slightly decreases before epoch 20. The relatively low model accuracy reflects that DPSGD adds a higher noise level to the gradients and harms the model utility, especially when the model is complicated.

\begin{figure*}[!t]
	\centering
	\subfigure[Performance of \sys~on LeNet-5 and MNIST.]{

		\begin{minipage}{0.33\textwidth}
        \centering
    \includegraphics[scale=0.75]{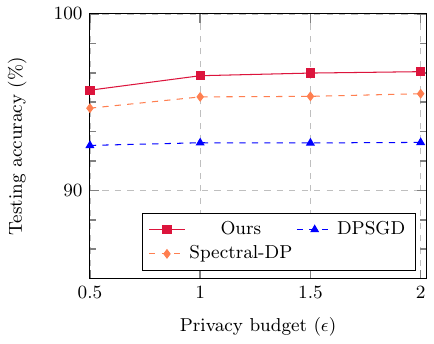}
		\end{minipage}

		\begin{minipage}[c]{0.33\textwidth}
  \includegraphics[scale=0.75]{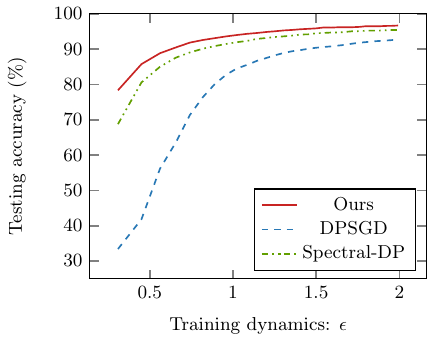}
		\end{minipage}

		\begin{minipage}[c]{0.33\textwidth}
    \includegraphics[scale=0.75]{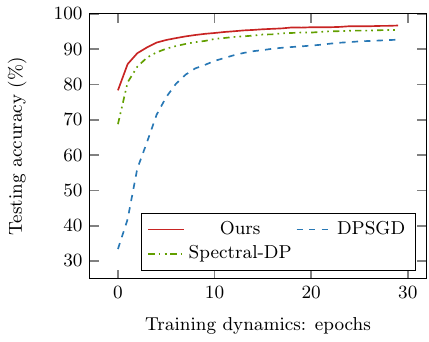}
		\end{minipage}
  }\label{fig:lenet}

 \subfigure[Performance of \sys~on ResNet-20 and CIFAR-10.]{
		\begin{minipage}[c]{0.33\textwidth}
      \includegraphics[scale=0.75]{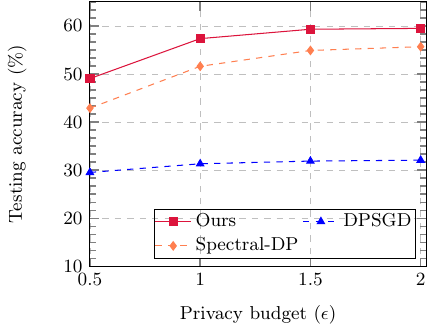}
		\end{minipage}
	
		\begin{minipage}[c]{0.33\textwidth}
        \includegraphics[scale=0.75]{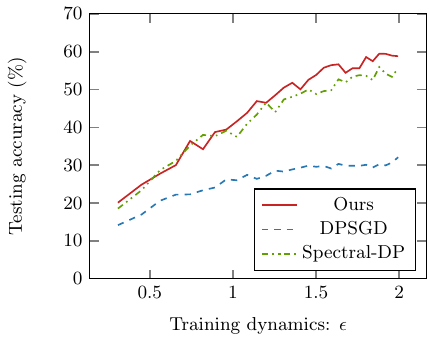}
		\end{minipage}
		
		\begin{minipage}[c]{0.33\textwidth}
          \includegraphics[scale=0.75]{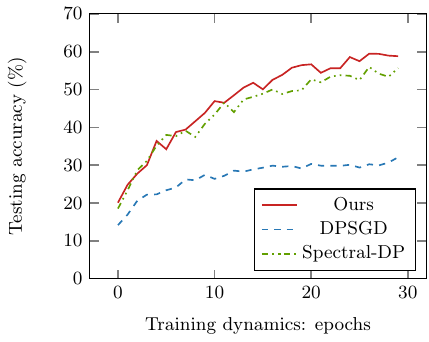}
		\end{minipage}
	}\label{fig:resnet}

 \subfigure[Performance on of \sys~on AlexNet and CIFAR-10.]{
		\begin{minipage}[c]{0.33\textwidth}
            \includegraphics[scale=0.75]{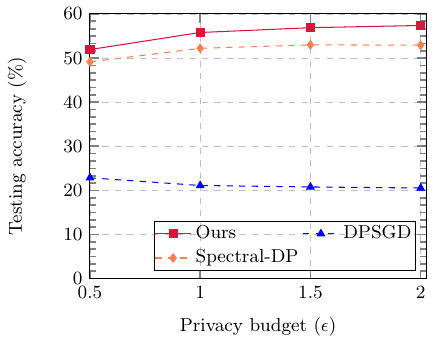}
		\end{minipage}
	
		\begin{minipage}[c]{0.33\textwidth}
              \includegraphics[scale=0.75]{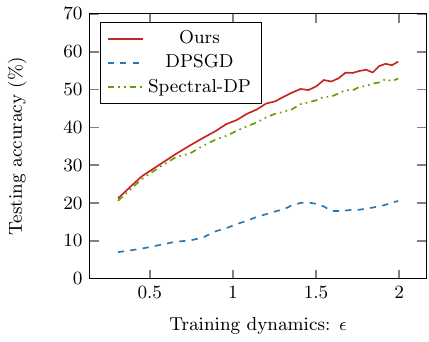}
		\end{minipage}
		
		\begin{minipage}[c]{0.33\textwidth}
                \includegraphics[scale=0.75]{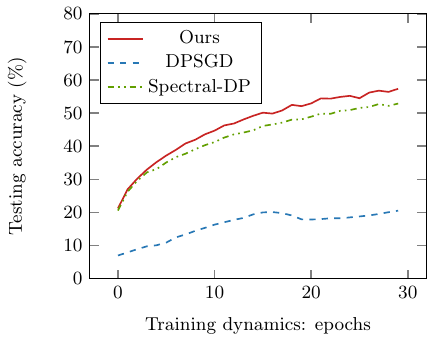}
		\end{minipage}
	}\label{fig:alexnet}

  \subfigure[Performance on of \sys~on Model-3 and CIFAR-10.]{
		\begin{minipage}[c]{0.33\textwidth}
            \includegraphics[scale=0.75]{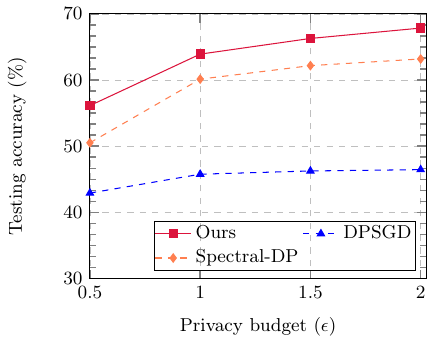}
		\end{minipage}
	
		\begin{minipage}[c]{0.33\textwidth}
              \includegraphics[scale=0.75]{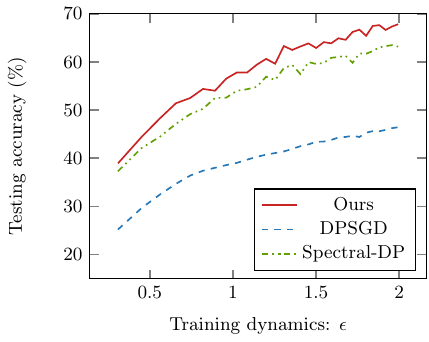}
		\end{minipage}
		
		\begin{minipage}[c]{0.33\textwidth}
                \includegraphics[scale=0.75]{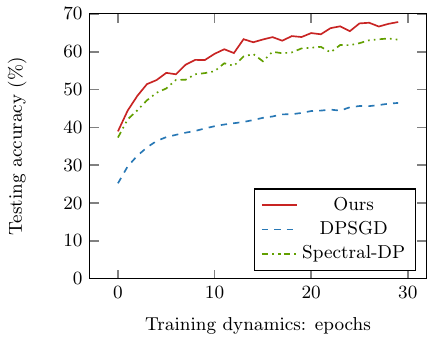}
		\end{minipage}
	}\label{fig:model3}
 \caption{Relationship between the testing accuracy and $\epsilon$ (left), and training dynamics when $\epsilon=2$ (middle and right).}
 \label{fig:acc}
\end{figure*}

In all, the experimental results in this section indicate that our \sys~can maintain the best model utility across various models and datasets. The high model accuracy conforms to our theoretical results, that our approach requires a lower noise scale compared to the baselines without discarding any gradient information. 

\subsection{Transfer Learning}
\label{subsec:exp:transfer}

We present the model accuracy of \sys~and baselines in the transfer learning settings. We adopt the ResNet-18 as the model architecture and evaluate our scheme in two different settings. ResNet-18 is another representative model architecture in the ResNet architectures, whose concrete structure is shown in Table \ref{tab:arch:resnet18} in Appendix \ref{app:stru:res18}. More concretely, we first fix the entire network of ResNet-18 and only re-train the last convolutional layer, which is denoted as ResNet-18-1Conv. Then we retrain the last two convolutional layers and denote it as ResNet-18-Conv2. For transfer learning, we set the learning rate as 0.001.

\begin{table*}[]
  \caption{Comparisons of accuracy (\%) in transfer learning on CIFAR-10 and CIFAR-100. }
    \label{tab:acc_transfer}
  \centering
  \resizebox{\textwidth}{!}{
\begin{tabular}{l|c|c|c|c|c|c|c|c|c|c|c|c}
\toprule
 \textbf{Datasets} & \multicolumn{6}{c|}{CIFAR-10} & \multicolumn{6}{c}{CIFAR-100} \\
 \midrule
  \textbf{Transfer settings} & \multicolumn{3}{c|}{ResNet18-1Conv} & \multicolumn{3}{c|}{ResNet18-2Conv} & \multicolumn{3}{c|}{ResNet18-1Conv} & \multicolumn{3}{c}{ResNet18-2Conv} \\
 \midrule
 \textbf{Privacy budget}& $\epsilon=0.5$ & $\epsilon=0.1$& $\epsilon=2$ & $\epsilon=0.5$ & $\epsilon=0.1$& $\epsilon=2$ & $\epsilon=0.5$ & $\epsilon=0.1$& $\epsilon=2$ & $\epsilon=0.5$ & $\epsilon=0.1$& $\epsilon=2$ \\
 \midrule
Plain & \multicolumn{3}{c|}{80.27} & \multicolumn{3}{c|}{86.59} & \multicolumn{3}{c|}{52.74} & \multicolumn{3}{c}{59.32} \\
\midrule
DPSGD \cite{dpsgd} & 63.87 & 65.98 & 67.05 & 62.51 & 65.7 & 67.64 & 19.41 & 26.89 & 30.44 & 16.11 & 23.25 & 27.57 \\
\midrule
Spectral-DP \cite{spectral} & 70.22 & 74.18 & 76.48 & 68.33 & 75.11 & 79.44 & 26.78 & 37.36 & 43.46 & 22.65 & 34.55 & 44.02 \\
\midrule
\sys~\textbf{(Ours)} & \textbf{70.57} & \textbf{74.76} & \textbf{76.83} & \textbf{69.96} & \textbf{76.77} & \textbf{80.86} & \textbf{27.13} & \textbf{37.7} & \textbf{44.14} & \textbf{23.96} & \textbf{35.96} & \textbf{45.12}\\
\bottomrule
\end{tabular}
}
\end{table*}

The model accuracy of \sys~and baseline methods under the transfer learning settings on CIFAR-10 and CIFAR-100 is shown in Table \ref{tab:acc_transfer}. \textbf{The results illustrate that our \sys~also exhibits better performance than the existing baselines in transfer learning.} 
For instance, \sys~achieves 76.83\% model accuracy when $\epsilon=2$ under ResNet-18-1Conv setting, which is 0.35\% and 9.78\% higher than Spectral-DP and DPSGD, respectively. Similarly, under the ResNet-18-2Conv setting, our method obtains 69.96\%, 76.77\%, and 80.86\% model accuracy when $\epsilon=0.5, 1$, and 2. Compared to the baseline approaches, our scheme is 1.42\% to 1.66\% higher than Spectral-DP and 7.45\% to 13.22\% higher than DPSGD. 
On CIFAR-100, the results exhibit a similar trend. For instance, compared to the model accuracy of the plain models (i.e., without DP perturbation), our approach is 4.56\% lower when $\epsilon=2$ on ResNet18-1Conv and CIFAR-10, which is 0.35\% and 9.78\% higher than Spectral-DP and DPSGD, respectively. Overall, the performance gain under transfer learning is less significant than training from scratch, which conforms to the anticipation that the majority of the model layers are fixed without re-training. 
Nevertheless, the results demonstrate that our \sys~can attain a better model utility than existing methods under the same privacy budget in transfer learning. The consistently better results derive from the fact that our method requires the lowest noise level without deleting any gradient information. 

\subsection{Ablation Study}
\label{subsec:exp:ablation}

\begin{figure}
    \centering
    \includegraphics[width=.9\columnwidth]{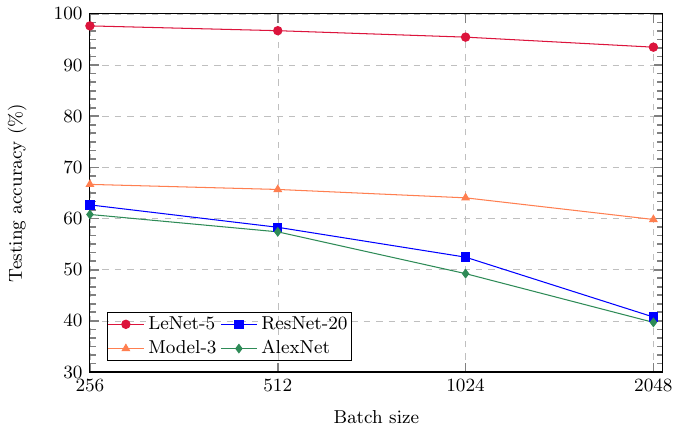}
    \caption{Relationship between the accuracy and batch size.}
    \label{fig:batchsize}
\end{figure}

\noindent\textbf{Batch Size.}
The results of batch size v.s. model accuracy when $\epsilon=2$ is shown in Figure \ref{fig:batchsize}. It can be seen from the figure that the model utility gradually decreases when the batch size grows larger. For instance, on LeNet-5, \sys~obtains 97.67\%, 96.73\%, 95.47\%, and 93.51\% model accuracy when batch size equals 256, 512, 1024, and 2048, respectively. The accuracy decreases by 4.16\% with batch size increases. The accuracy of Model-3 shows a similar trend, which decreases by 6.86\% when the batch size is 2048. The variation of accuracy is more significant on ResNet-20 and AlexNet, which exhibit 21.92\% and 21.1\% change, respectively. Note that a smaller batch size would lead to slower training procedures. In our paper, we set the batch size as 500 to better compare with the baselines, which is also close to the optimal setting. 

\noindent\textbf{Learning Rate.}
The relationship between the learning rate and the model accuracy is shown in Figure \ref{fig:abl_lr}. Overall, the model accuracy decreases as the learning rate grows larger. Among all the models, LeNet-5 achieves the highest accuracy and is relatively stable. When the learning rate increases to 0.01, the model attains 96.66\% utility, which is only 0.07\% lower than the learning rate equals 0.001. With the learning rate getting larger to 0.01, the model accuracy decreases to 80.18\% on LeNet-5. 
The other three models share a similar trend, that the accuracy drops to around 40\% when the learning rate is 0.025, increases slightly back, and then gradually decreases as the learning rate grows. 
In our experiments, we set the learning rate as 0.01 for the optimal results, which is also selected as the preferable setting for Spectral-DP.

\begin{figure}
    \centering
    \includegraphics[width=.9\columnwidth]{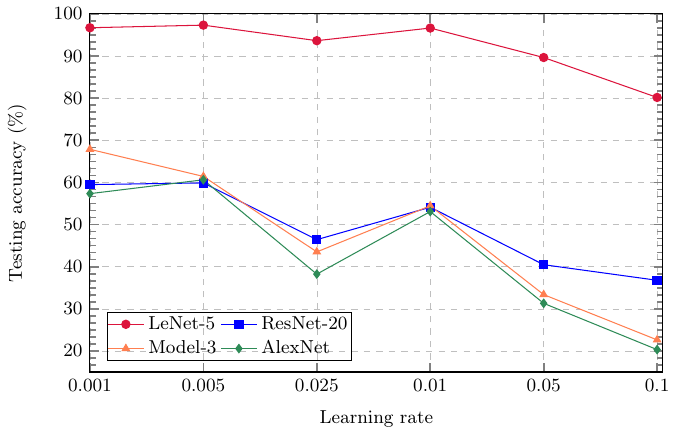}
    \caption{Relationship between the accuracy and learning rate.}
    \label{fig:abl_lr}
\end{figure}

\begin{table*}[htbp]
  \caption{Model accuracy (\%) with different values of clipping norm and privacy budgets.}
		\label{tab:abl_clip}
\centering
	\begin{tabular}{l|c|c|c|c|c|c|c|c}
		\toprule
           Dataset & \multicolumn{2}{|c}{MNIST} & \multicolumn{2}{|c}{CIFAR-10} & \multicolumn{2}{|c}{CIFAR-10} & \multicolumn{2}{|c}{CIFAR-10}\\
           \midrule
           Models & \multicolumn{2}{|c}{LeNet-5} & \multicolumn{2}{|c}{ResNet-20} & \multicolumn{2}{|c}{Model-3} & \multicolumn{2}{|c}{AlexNet} \\
            \midrule
            Clipping Norm & 0.1  & 1 & 0.1 & 1 & 0.1  & 1 & 0.1 & 1\\
            \midrule
            $(\epsilon, \delta)=(0.5, 10^{-5})$ & 95.47 & \textbf{95.68} & 43.85 & \textbf{49.06} & \textbf{60.2} & 56.15 & \textbf{52.25} & 51.93\\
            $(\epsilon, \delta)=(1, 10^{-5})$ & 96.43 & \textbf{96.5} & 44.37 & \textbf{57.36} & 60.56 & \textbf{63.93} & 52.49 & \textbf{55.8}\\
            $(\epsilon, \delta)=(1.5, 10^{-5})$ & 96.42 & \textbf{96.65} & 44.41 & \textbf{59.31} & 60.31 & \textbf{66.3} & 52.44 & \textbf{56.89}\\
            $(\epsilon, \delta)=(2, 10^{-5})$ & 96.4 & \textbf{96.73} & 45.38 & \textbf{59.49} & 60.2 & \textbf{67.87} & 52.44 & \textbf{57.39}\\
			\bottomrule
		\end{tabular}
\end{table*}

\noindent\textbf{Clipping Norm.}
We present the accuracy with different values of clipping norm and privacy budgets in Table \ref{tab:abl_clip}. For each model, we evaluate \sys~with clipping norm as 1 and 0.1. It can be seen from the table that for the majority of the experiments, clipping norm equals one is a better choice. For example, \sys~achieves 96.73\% model accuracy on LeNet-5 when clipping norm as 1 and privacy budget as $(2, 10^{-5})$, while it obtains 96.4\% when clipping norm is 0.1. The gap becomes larger on ResNet-20 and Model-3, e.g., $\mathsf{acc}=59.45$ when $\mathsf{cl}=1$ and $\epsilon=1.5$ while $\mathsf{acc}=44.41$ when $\mathsf{cl}=0.1$, causing a accuracy decrease of 15.04\%. Although the accuracy on Model-3 when $\mathsf{cl=0.1}$ and $ \epsilon=0.5$ is 4.05\% higher than $\mathsf{cl=1}$, \sys~performs better when the clipping norm is set to 1 in general. 

\section{Related Work}
\label{sec:relatedwork}




\subsection{Differential Privacy in Deep Learning}
\label{subsec:rl:dpfordl}

Regarding different privacy goals, DP can be adapted to various stages of the deep learning models, including the input data, the object functions, data labels, the well-trained model parameters (i.e., outputs), and the training procedures. 

\noindent\textbf{Input Perturbations.} DP that is applied to the input data of the deep learning model is equivalent to the training with sanitized data. More concretely, the noise is injected into the input data so that the global distribution is retained while information on each piece of data can be concealed \cite{gopi2020differentially, mcmahan2017learning, kifer2014pufferfish, andres2013geo}. This strategy is usually adopted when there is sensitive information contained in the training data in order to defend against the model inference attack. For instance, Google proposed RAPPOR \cite{erlingsson2014rappor}, which permits the developer to privately collect the user's data. Triastcyn \emph{et al.} brought forward a method that can generate a synthesis dataset with the same distribution as the original data yet preserve the data privacy \cite{triastcyn2018generating}. There are also quantities of work \cite{298240} that combines the input perturbation with other techniques to achieve more complicated goals. 

\noindent\textbf{Object Function Perturbations.} Methods in this category involve the noise addition to the object function of the deep learning models. Phan \emph{et al.} propose the deep private auto-encoder (dPA) \cite{phan2016differential}, which realizes differentially private characteristics based on the functional mechanism. In \cite{phan2017preserving}, the authors design an Adaptive Laplace Mechanism (AdLM), which provides a novel scheme to perturb the object function. More recently, Iyengar \emph{et al.} establish a benchmark using high-throughout data based on a more practical objective perturbation function \cite{iyengar2019towards}. To better calculate the function sensitivity for the noise intensity, some works propose to use the objective function perturbation upon an approximate convex polynomial function \cite{phan2016differential, 287254}. The above-mentioned work can be usually adopted before the network training procedure. 

\noindent\textbf{Label Perturbations}. The DP mechanism can also be exerted on the label level of the dataset. One of the typical applications of label perturbation is to inject noise into the transfer phase in the teacher-student framework \cite{zhao2018distributed}. Papernot \emph{et al.} propose PATE \cite{papernot2017semisupervised}, which is a teach-student model that permits transfer learning with differentially private characteristics. This method is further extended to large-scale applications with a novel noise aggregation scheme \cite{papernot2018scalable}. 
Based on this, some works adopt PATE as the discriminator of the Generative Adversarial Networks (GANs) to achieve DP guarantees.

\noindent\textbf{Output Perturbations.} Similar to the objective function perturbation, the output perturbation is also designed for deep learning tasks with convex objective functions \cite{mireshghallah2020privacy, 298182, 10.1145/3576915.3616592}. Rastogi \emph{et al.} propose to inject noise to the output to accurately answer high sensitivity queries \cite{10.1145/1559795.1559812}. There are also a bunch of works that utilize DP-based output to preserve the model privacy of clients in the federated learning setting \cite{liu2024cross, 10.1145/3576915.3623193, 287324}. 

\noindent\textbf{Gradient Perturbation.} Despite the effectiveness of the above methods, they cannot defend against adversaries that have access to the training dynamics. Moreover, directly adding noise to the well-trained parameters would devastate model utility. Thus to address those capability limits, in 2016, Abadi \emph{et al.} proposed to inject the noise into the entire training procedure \cite{dpsgd}. The designed scheme, called DPSGD, realizes efficient and DP-guaranteed SGD algorithms for general deep learning models. Since then, there have been quantities of work brought forward in this region \cite{294621}. Yu \emph{et al.} propose the gradient embedding perturbation (GEP) which projects the weight matrix to a low-dimensional domain \cite{yu2021do}. Several works utilize mapping functions to project or decompose the gradient weights to a different space for better performance \cite{yu2021large, nasr2020improving}. There are also approaches that focus on optimizing the utility of DPSGD from various aspects, such as controlling the clipping norm \cite{papernot2021tempered}, minimizing the trainable parameters \cite{luo2021scalable}, designing DP-friendly model architectures \cite{cheng2022dpnas}, privacy budget allocation \cite{asi2021private, yu2019differentially}, and selecting suitable data features \cite{tramer2021differentially}. These works follow the fundamental design of DPSGD and present enhancing methods as an auxiliary improvement. 
In 2023, Feng \emph{et al.} propose Spectral-DP \cite{spectral}, which modifies the conceptual foundations of DPSGD by injecting noise in the Fourier domain. As an alternative to DPSGD, Spectral-DP permits the deep learning models to retain higher accuracy. However, Spectral-DP relies highly on the filtering ratio and inevitably harms the original gradient information. 

\subsection{Domain Transformations in DP}

In the more general context of DP, some previous works proposed that utilizing the sparse distribution of the vectors in the transferred domain can improve the performance of DP. A plethora of work adopts the Discrete Wavelet Transformation to approximate the query processing or facilitate the query optimization \cite{chakrabarti2001approximate, vitter1999approximate, garofalakis2005wavelet}. Similarly, Xiao \emph{et al.} propose Privelet, which utilizes DWT in the database to provide accurate results for range-count queries \cite{xiao2010differential}. Another line of work makes use of the FFT to the database and achieves promising results \cite{cormode2018marginal, yamamoto2022efficient}. For instance, Rastogi \emph{et al.} design a Fourier perturbation algorithm to ensure DP for time-series data despite the presence of temporal correlation \cite{FFT1}.
The success of these methods suffices to illustrate the effectiveness of the domain transformation, which has not been extensively adopted in deep learning training algorithms. Enlightened by the observations, in this work, we also compute the gradients in each iteration in the Fourier domain, which further enables our noise reduction approach.

\section{Limitations}
\label{app:limitation}

Currently, the implementation of \sys~is straightforward without any thread-level parallelization. Consequently, the training procedure is relatively slow. This issue can be solved by integrating the implementation to the $\mathsf{octpus}$ library. 
Furthermore, \sys~adopts the Gaussian mechanism to apply the noise. There might be a better way to exert the noise to the gradients. Finally, this work mainly focuses on the traditional deep learning models. The privacy issues of training or fine-tuning other architectures (e.g., large language models) will be researched in the future.

\section{Conclusion}

In this paper, we propose a robust and provably secure approach for differentially private training called \sys. Our method computes the model gradients in the frequency domain via FFT and reduces the noise scale using real-part selection. At the core of our method lies a new gradient-preserving noise reduction strategy. Compared to the previous methods, our \sys~only requires half of the noise scale than DPSGD while preserving all the gradient information. We analyze the performance of \sys~theoretically which demonstrates its effectiveness. 
The experimental results show that our method works consistently better than the baselines under both direct training and transfer learning settings. 


\section*{Ethics Considerations}
This work addresses the development and improvement of differentially private training methods for deep learning models. 
The methodology and experiments presented in this paper rely on publicly available datasets, ensuring that no private or sensitive information was used or exposed. Furthermore, the findings and techniques are general in nature and are intended to advance the state of the art in privacy-preserving machine learning. They do not pose direct or immediate ethical risks, as the focus is on the framework and principles rather than specific implementations.

\section*{Open Science}

Based on the experimental settings, we fully implement our proposed method. 
Our code is open-sourced and available at
\begin{center}
    \href{https://anonymous.4open.science/r/FReDP-4CB2}{https://anonymous.4open.science/r/FReDP-4CB2}
\end{center}


\bibliographystyle{plain}
\bibliography{ref}

\appendix


\section{Detailed Proofs}

\subsection{Proof of Theorem \ref{theo:extendmain}}
\label{app:proof_theo_2}

\noindent\textbf{Theorem 2.}
    Suppose the noise is $\tau = \tau_x + i \cdot \tau_y$. The real and imaginary part of $\tau$ satisfy: $\tau_x \sim \mathcal{N}(0, a^2 S^2 \sigma^2)$ and $\tau_y \sim \mathcal{N}(0, b^2 S^2 \sigma^2)$, where $\sqrt{a^2+b^2}=1$. Then the results in Theorem \ref{theo:main} still hold.

\begin{proof}
    Let the vector as $\mathbf{V}=(V_1, \cdots, V_N)$, and $\mathbf{T}= \mathcal{F}(\mathbf{V})$, such that $T_j=\frac{1}{\sqrt{N}}\sum_{k=1}^N V_k\exp(-jki2\pi/N)$. Then the noise addition as described in Section \ref{subsec:method:approach} is $\hat{\mathbf{T}}=\mathbf{T}+\tau$. Thus, performing the IFFT on the noised vector in the frequency domain can be represented as 
    \begin{equation}
    \begin{split}
    \mathcal{F}^{-1}(\hat{\mathbf{T}})
			&=\mathcal{F}^{-1}(\mathbf{T})+\mathcal{F}^{-1}(\tau)=\mathbf{V}+\mathcal{F}^{-1}(\tau)
		\end{split}
    \end{equation}
    where 
    \begin{equation}
        \begin{split}
            [\mathcal{F}^{-1}(\tau)]_j&=\frac{1}{\sqrt{N}}\sum_{k=1}^N\tau_k\exp(jki2\pi/N)\\
            &=\frac{1}{\sqrt{N}}\sum_{k=1}^N\tau_k(\cos(\frac{2\pi}{N}kj)+i\sin(\frac{2\pi}{N}kj))
        \end{split}
    \end{equation} 
    Then we have
    \begin{equation}
        \begin{split}
            &\mathcal{R}([\mathcal{F}^{-1}(\tau)])_j\\
            =&\frac{1}{\sqrt{N}}\sum_{k=1}^{N}\tau_x\cos(\frac{2\pi}{N}kj)+\tau_y\sin(\frac{2\pi}{N}kj)\\
            \sim&\mathcal{N}(0,\frac{1}{N}S^2\sigma^2 (a^2 \sum_{k=1}^N\cos^2(\frac{k}{N}2\pi j)+b^2 \sum_{k=1}^N\sin^2(\frac{k}{N}2\pi j)))\\
            =&\mathcal{N}(0,\frac{1}{N}S^2\sigma^2(a^2\frac{N}{2}+b^2\frac{N}{2}))\\
            =&\mathcal{N}(0,\frac{1}{2}S^2\sigma^2)
        \end{split}
    \end{equation}
Thus,
\begin{equation}
    \mathsf{GRe}(\mathcal{F}(\mathbf{V})) = \mathbf{V}+\mathcal{R}(\mathcal{F}^{-1}(\tau))
    \sim \mathbf{V}+\mathcal{N}(0,\frac{1}{2} S^2 \sigma^2)
\end{equation}
which completes the proof.
\end{proof}

\subsection{Proof of Theorem \ref{theo:dpsingle}}
\label{app:proof_theo_3}

\noindent\textbf{Theorem 3.}
In Algorithm 1 (i.e., corresponding to our method overview), the output $\mathsf{FRe}(\mathbf{G})$ is $(\epsilon, \delta)$ differential private if we choose $\sigma$ to be $\sqrt{2\log (1.25/\delta)}/\epsilon$.  

\begin{proof}
    Suppose $\sigma=\sqrt{2\log (1.25/\delta)}/\epsilon$, then $\tilde{\mathbf{G}}$ (i.e., the vector after noise addition) follows a normal distribution with mean 0 and variance $2 S^2 \log (1.25/\delta)/\epsilon^2$, where $S$ is the sensitivity. Deriving from Theorem 3.22 in \cite{dwork2014algorithmic}, $\tilde{\mathbf{G}}$ satisfies $(\epsilon, \delta)$-DP. Finally, we then adopt the IFFT and imaginary part deletion to $\tilde{\mathbf{G}}$, which are all post-processing procedures, i.e., the computation is independent of the data. Thus, the result $\mathsf{GRe}(\mathbf{G})$ satisfies $(\epsilon, \delta)$-DP.
\end{proof}

\subsection{Proof of Corollary \ref{coro:conv}}
\label{app:proof_coro_1}

\noindent\textbf{Corollary 1.}
Suppose $\Gamma =\{\gamma_{m,n}\}_{m,n \in [0,N-1]}$ is a matrix of noise vector in frequency domain where $\gamma_{m,n} \sim \mathcal{N}(0, \sigma^2 S^2)$. Then the elements in $\mathcal{R}(\mathcal{F}^{-1}(\Gamma))$ follows a normal distribution $\mathcal{N}(0, \frac{1}{2}\sigma^2 S^2)$.

\begin{proof}
    Denote $\Upsilon=\{\upsilon_{k,l}\}_{k,l \in [0, N-1]}$ as the corresponding vector of $\Gamma$ in the time domain, such that the $(k,l)$-th element in $\mathcal{F}^{-1}(\Gamma)$ is
    \begin{equation}
    \begin{aligned}
        &\upsilon_{k,l} = \frac{1}{N} \sum_{m=0}^N \sum_{n=0}^N \gamma_{m,n} \cdot e^{\frac{2\pi i}{N} \cdot (mk + nl)} \\
        &= \frac{1}{N} \sum_{m=0}^N \sum_{n=0}^N \gamma_{m,n} \cdot \left( \cos (\frac{2\pi}{N}(mk +nl)) + i \sin (\frac{2\pi}{N}(mk +nl))\right)\\
        &\sim \mathcal{N} (0, \sum_{m=0}^N \sum_{n=0}^N \frac{1}{N^2} \cos^2 (\frac{2\pi}{N}(mk +nl)) \sigma^2 S^2) \\
        & + i \cdot \mathcal{N} (0, \sum_{m=0}^N \sum_{n=0}^N \frac{1}{N^2} \sin^2 (\frac{2\pi}{N}(mk +nl)) \sigma^2 S^2)
    \end{aligned}
    \end{equation}
    where the last $\sim$ derives from the properties of the normal distribution. Thus, we have
    \begin{equation}
    \begin{aligned}
        \mathcal{R}(\upsilon_{k,l}) &\sim \mathcal{N} (0, \sum_{m=0}^N \sum_{n=0}^N \frac{1}{N^2} \cos^2 (\frac{2\pi}{N}(mk +nl)) \sigma^2 S^2) \\
        &= \mathcal{N}(0, \frac{1}{2}\sigma^2 S^2)
    \end{aligned}
    \end{equation}
    which completes the proof. 
\end{proof}

\subsection{Proof of Corollary \ref{coro:train}}
\label{app:proof_coro_2}

\noindent\textbf{Corollary 2.} The DL training algorithm we proposed in Algorithm \ref{alg:trainingalgorithm} achieves $((T_e \cdot N/B)\epsilon + \frac{\log(1/\delta)}{\alpha-1}, \delta)-DP$ if $\sigma = \sqrt{2\log(1.25/\delta)}/\epsilon'$, where $\epsilon' = \epsilon + \frac{\log(1/\delta)}{\alpha-1}$.

\begin{proof}
    Following Theorem 3.22 \cite{dwork2014algorithmic}, the Gaussian mechanism in Algorithm \ref{alg:trainingalgorithm} is $(\epsilon', \delta)$ DP, which is $(\alpha, \epsilon)$-RDP due to Proposition \ref{pro:concate}. 
    Suppose the size of the dataset is $N$ with batch size $B$ and training epoch $T_e$. Thus there are in total $T_e \cdot \frac{N}{B}$ times of noise addition. Due to the property of the RDP mechanism, Algorithm 4 guarantees a $(\alpha, T_e \cdot \frac{N}{B} \epsilon)$-RDP, which is equivalent to a $(T_e \cdot \frac{N}{B} \epsilon + \frac{\log 1/\delta}{\alpha-1}, \delta)$-DP.
\end{proof}





\section{Detailed Model Structures}
\label{app:structure}

\subsection{Structure of LeNet-5} \label{app:stru:lenet5} The detailed model structure of LeNet-5 is shown in Table \ref{tab:arch:lenet}.

\begin{table}[]
    \caption{Architecture of LeNet-5 model.}
    \label{tab:arch:lenet}
    \begin{center}
    \begin{tabular}{c|c}
    \toprule
    Layer & Model Parameters \\
    \midrule
      Convolution   &  6 kernels of size $5\times 5$\\
      Max-pooling   & kernel size $2\times 2$\\
      Convolution & 6 kernels of size $5\times 5$\\
      Max-pooling & kernel size $2\times 2$\\
      Fully-connection $\times$3 & 120, 84, and 10 units, respectively\\ 
 \bottomrule
    \end{tabular}
    \end{center}
\end{table}

\subsection{Structure of ResNet-20} \label{app:stru:resnet20} The detailed structure of ResNet-20 is shown in Table \ref{tab:arch:resnet20}.

\subsection{Structure of Model-3}
\label{app:stru:model3}

The model structure of Model-3 is demonstrated in Table \ref{tab:arch:model3}.

\subsection{Structure of ResNet-18}
\label{app:stru:res18}

The model structure of ResNet-18 is shown in Table \ref{tab:arch:resnet18}.

\subsection{Structure of AlexNet}
\label{app:stru:alex}

The detailed model structure of AlexNet is shown in Table \ref{tab:arch:alexnet}.

\begin{table}[]
    \caption{Architecture of ResNet-20 model.}
    \label{tab:arch:resnet20}
    \begin{center}
    \begin{tabular}{c|c}
    \toprule
    Layer & Model Parameters \\
    \midrule
      Convolution   &  7 kernels of size $3\times 3$\\
      Avg-pooling   & kernel size $2\times 2$\\
      Convolution & 6 kernels of size $3\times 3$\\
      Avg-pooling & kernel size $2\times 2$\\
      Convolution & 6 kernels of size $3\times 3$\\
      Avg-pooling & kernel size $2\times 2$\\
      Fully-connection & 64 units\\ 
\bottomrule
    \end{tabular}
    \end{center}
\end{table}

\begin{table}[]
    \caption{Architecture of Model-3 model.}
    \label{tab:arch:model3}
    \begin{center}
    \begin{tabular}{c|c}
    \toprule
    Layer & Model Parameters \\
    \midrule
      Convolution   &  2 kernels of size $3\times 3$\\
      Max-pooling   & kernel size $2\times 2$\\
      Convolution & 2 kernels of size $3\times 3$\\
      Max-pooling & kernel size $2\times 2$\\
        Convolution & 2 kernels of size $3\times 3$\\
      Max-pooling & kernel size $2\times 2$\\
      
      Fully-connection $\times$2 & 120 and 10 units, respectively\\ 
\bottomrule
    \end{tabular}
    \end{center}
\end{table}

\begin{table}[]
    \caption{Architecture of ResNet-18 model.}
    \label{tab:arch:resnet18}
    \begin{center}
    \begin{tabular}{c|c}
    \toprule
    Layer & Model Parameters \\
    \midrule
      Convolution   &  1 kernels of size $7\times 7$\\
      Max-pooling   & kernel size $3\times 3$\\
      Convolution & 4 kernels of size $3\times 3$\\
      Max-pooling & kernel size $2\times 2$\\
      Convolution & 4 kernels of size $3\times 3$\\
      Max-pooling & kernel size $2\times 2$\\
      Convolution & 4 kernels of size $3\times 3$\\
      Max-pooling & kernel size $2\times 2$\\
      Convolution & 4 kernels of size $3\times 3$\\
      Max-pooling & kernel size $2\times 2$\\
      
      Fully-connection & 10 units\\ 
\bottomrule
    \end{tabular}
    \end{center}
\end{table}

\begin{table}[]
\caption{Architecture of AlexNet model.}
    \label{tab:arch:alexnet}
    \begin{center}
    \begin{tabular}{c|c}
    \toprule
    Layer & Model Parameters \\
    \midrule
      Convolution   &  1 kernels of size $3\times 3$\\
      Max-pooling   & kernel size $2\times 2$\\
      Convolution & 1 kernels of size $3\times 3$\\
      Max-pooling & kernel size $2\times 2$\\
      Convolution & 1 kernels of size $3\times 3$\\
      Max-pooling & kernel size $2\times 2$\\
      Convolution & 1 kernels of size $3\times 3$\\
      Max-pooling & kernel size $2\times 2$\\
      Convolution & 1 kernels of size $3\times 3$\\
      Max-pooling & kernel size $2\times 2$\\
      
      Fully-connection $\times$3 & 120, 84, and 10 units, respectively\\ 
\bottomrule
    \end{tabular}
    \end{center}
\end{table}

\end{document}